\documentclass[11pt]{amsart}

\usepackage{amsfonts}
\usepackage{graphicx}

\usepackage{amssymb}
\usepackage{amsmath}
\usepackage{amscd}
\usepackage{mathrsfs}
\usepackage{stmaryrd}
\usepackage{bbm}
\usepackage{longtable}
\usepackage{txfonts}
\usepackage{wrapfig}
\usepackage{picins}
\usepackage{float}
\usepackage{xypic}
\usepackage{dsfont}
\usepackage{xcolor}
\usepackage{color}
\usepackage[colorlinks,linkcolor=blue,anchorcolor=blue,
citecolor=blue]{hyperref}
\usepackage{hyperref}
\allowdisplaybreaks

\newtheorem{theorem}{Theorem}[subsection]
\newtheorem{prop}[theorem]{Proposition}

\newtheorem{lem}[theorem]{Lemma}
\newtheorem{coro}[theorem]{Corollary}

\newtheorem{thm}[theorem]{Theorem}

\begin{document}
\setlength{\oddsidemargin}{0cm}
\setlength{\evensidemargin}{0cm}

\title{Some Varieties of Lie Rings}

\author{Yin Chen}

\address{
School of Mathematics and Statistics, Northeast Normal University,
Changchun 130024, P.R. China and
Chern Institute of Mathematics, Nankai University, Tianjin 300071, P.R. China
}

\email{ychen@nenu.edu.cn}

\author{Runxuan Zhang}

\address{School of Mathematics and Statistics, Northeast Normal University,
 Changchun 130024, P.R. China}

\email{zhangrx728@nenu.edu.cn}

\date{\today}

\def\shorttitle{Some Varieties of Lie Rings}

\begin{abstract}
In this paper, several theorems of Macdonald \cite{Mac1961,Mac1962} on the varieties of nilpotent groups will be generalized  to the case of Lie rings.
We consider  three varieties of Lie rings of any characteristic associated  with some equations (see Eqs. (\ref{eq:1.1})-(\ref{eq:1.3}) below).
We prove  that each Lie ring in variety $(\ref{eq:1.1})$ is nilpotent of exponent at most $n+2$; if $L$ is a Lie ring in variety $(\ref{eq:1.2})$, then
 $L^{2}$ is nilpotent of exponent at most $n+1$; and each Lie ring in variety $(\ref{eq:1.3})$ is solvable of length at most $n+1$. Finally, we also discuss
 some varieties of solvable Lie rings and the varieties of  Lie rings defined  by the properties of subrings.
\end{abstract}



\maketitle
\baselineskip=15pt


\section{Introduction}
\setcounter{equation}{0}
\renewcommand{\theequation}
{1.\arabic{equation}}

\setcounter{theorem}{0}
\renewcommand{\thetheorem}
{1.\arabic{theorem}}

A \textit{class of groups} $\mathfrak{X}$ is a class whose members are groups and which has the following properties: (1)
$\mathfrak{X}$ contains a group of order 1; (2) $H\simeq G\in \mathfrak{X}$ always implies that $H\in \mathfrak{X}$.
A \textit{variety of groups} is a class of groups which are defined by sets of equations.
For example, all abelian groups form the variety associated with the equation  $(x,y)=1$, where $(x,y)=xyx^{-1}y^{-1}$ is the commutator of elements $x$ and $y$. The study of varieties of groups is a subject with a long history.
In 1961, Macdonald \cite{Mac1961} proved that two kinds of varieties of groups associated with equations
\begin{eqnarray*}
(x,y_{1},\cdots,y_{n},x)&=&1;\\
{}(x,y_{1},\cdots,y_{n};x,y_{1},\cdots,y_{m})&=&1
\end{eqnarray*}  are both nilpotent, and the corresponding nilpotent exponents was also calculated.
In the same paper, Macdonald  pointed out that the analogous results may be valid for some varieties of Lie rings, but he did not give any proof. Later, a lot of papers addressing the varieties of finite nilpotent groups and solvable groups appeared, such as \cite{Kik1967,Lev1964,Lev1968,Mac1962,Mac1965}, in which  the attention was focused  on the subgroups of solvable groups with some restrictions.

As we know, there are a lot of similar concepts and results between finite group theory and finite-dimensional Lie algebras.
It is natural to bring some methods in finite group theory into the study of  Lie algebras.
Recently, Suanmali \cite{Sua2008} (or \cite{Sua2007}) used an analogous idea in the theory of group varieties to
investigate the varieties of  Lie algebras.
She considered the exponent bound problem  for some varieties of nilpotent Lie algebras, and extended Macdonald's results in \cite{Mac1961, Mac1962} to
finite-dimensional Lie algebras over a field of characteristic  not 2 and 3. Apart from this, very
little results seems to be known about the varieties of Lie algebras.
In this paper, we  follow basically the routines of Macdonald and Suanmali, and  reconsider the varieties of Lie rings.
We also obtain some analogues of Macdonald and Suanmali's results. But, the advantage of our results is that the restriction of characteristic and the structure of vector spaces are taken away.

Specifically,
our first purpose is to discuss the following three varieties of Lie rings of any characteristic associated  with equations
\begin{eqnarray}
(x,y_{1},\cdots,y_{n},x)&=&0;\label{eq:1.1}\\
(x,y;x_{1},x_{2},\cdots,x_{2n-1},x_{2n};x,z)&=&0;\label{eq:1.2}\\
(x,y;x_{1},x_{2},\cdots,x_{2n};x,z;y_{1},y_{2},\cdots y_{m})&=&0.\label{eq:1.3}
\end{eqnarray}
We shall extend some results  in \cite{Sua2008}, such as several corresponding Macdonald's theorems, to the varieties of Lie rings (Theorems \ref{th:2.1}, \ref{thm:2.3} and \ref{thm:2.4}).
The second purpose is to prove
that the variety $V_{n}$ associated with the equation
 $$(x_{1},\cdots,x_{n},x_{n+1},x_{n+2})=(x_{1},\cdots,x_{n},x_{n+2},x_{n+1})$$ is precisely the variety
 $V'_{n}$ associated with the equation
 $(x_{1},\cdots,x_{n};x_{n+1},x_{n+2})=0$; and  the variety $P_{n}$ associated with the equation
 $$(x_{1},\cdots,x_{i},x_{i+1},\cdots, x_{n})=(x_{1},\cdots,x_{i+1},x_{i},\cdots, x_{n})$$ is precisely the variety
 $P'_{n}$ associated with the equation
 $(x_{1},\cdots,(x_{i},x_{i+1}),\cdots, x_{n})=0.$ As a corollary, we obtain that each Lie ring $L$ in the variety  $P_{n}$ is solvable.
 We remark that the corresponding results in finite group theory are due to Levin (\cite{Lev1964}, Corollary 2.4) and Kikodze (\cite{Kik1967}, Corollary 2) respectively.
Moreover, we shall show that a Lie ring  $L$ satisfies
$
(x_{1},x_{2},\cdots,x_{n})=(x_{1},x_{\rho(2)},\cdots,x_{\rho(n)})
$
for all non-identical permutations $\rho$
of $\{2,\cdots,n\}$
if and only if $L$ is nilpotent of exponent at most $n-1$ (Theorem \ref{thm:2.8}).
In the last section, we also discuss
 some varieties of Lie ring defined  by properties of subrings.

To do this, we first need some preliminaries.

A nonassociative ring $L$ with the  multiplication $[-,-]$ is called a \textit{Lie ring}
if the following identities are satisfied:
\begin{eqnarray}
[x,y]+[y,x]&=&0~~(\textrm{skew symmetry});\\
{}[x,[y,z]]+[y,[z,x]]+[z,[x,y]]&=&0 ~~(\textrm{Jacobi identity}).
\end{eqnarray}
for all $x,y,z \in L$.
Let $x_{1},\cdots,x_{n},y_{1},\cdots,y_{m}$ be elements in a  Lie ring $L$, the \textit{commutator} $(x_{1},\cdots,x_{n})$ in $L$ is defined as
$$(x_{1},x_{2}):=[x_{1},x_{2}],\textrm{and } (x_{1},\cdots,x_{n}):=((x_{1},\cdots,x_{n-1}),x_{n}),\textrm{ for all } n>2.$$
Moreover, we define
$(x_{1},\cdots,x_{n};y_{1},\cdots,y_{m}):=((x_{1},\cdots,x_{n}),(y_{1},\cdots,y_{m})).$

Let $\rho$ be a permutation of $\{3,\cdots,k\}$. If
a Lie ring $L$ satisfies the equation  $$(x_{1},x_{2},x_{3},\cdots,x_{k})=(x_{1},x_{2},x_{\rho(3)},\cdots,x_{\rho(k)}),$$
we say that $L$ \textit{satisfies} $C(k,\rho)$.
 If $L$ satisfies $C(k,\rho)$ for all permutations $\rho$ of $\{3,\cdots,k\}$, we say that $L$ \textit{satisfies} $C(k)$.

For a  Lie ring $L$, we define the \textit{lower central series}
$$L^{1}:=L \supset L^{2}:=[L,L]\supset\cdots\supset L^{n}:=[L,L^{n-1}]\supset\cdots$$ for all $n=2,3,\cdots$.
A Lie ring $L$ is said to be \textit{nilpotent} if $L^{n+1}=\{0\}$ for some nonnegative integer $n$. If $L^{n+1}=\{0\}$ but $L^{n}\neq\{0\}$, then
we say that $L$ has \textit{nilpotent exponent} $n$.
We also define a different kind of powers of $L$,
$$L^{(0)}:=L \supset L^{(1)}:=[L,L]\supset\cdots\supset L^{(n+1)}:=[L^{(n)},L^{(n)}]\supset\cdots.$$ Every $L^{(n)}$ is an ideal of $L$ and
$L$ is said to be \textit{solvable} if $L^{(n+1)}=\{0\}$ for some nonnegative integer $n$. Moreover if $L^{(n+1)}=\{0\}$ but $L^{(n)}\neq\{0\}$, then
we say that $L$ has \textit{solvable length} $n$.
A  Lie ring $L$  is said to be \textit{abelian} if $L^{2}=\{0\}$; $L$ is said to be \textit{metabelian} if $L^{2}$ is abelian.

A \textit{class of Lie rings} $\mathfrak{L}$ is a class whose members are Lie rings and which has the following properties: (1)
$\mathfrak{L}$ contains the zero ring; (2) $L\simeq E\in \mathfrak{L}$ always implies that $L\in \mathfrak{L}$.
A \textit{variety of Lie rings} is a class of Lie rings which are defined by sets of equations. For example, all abelian Lie rings form the variety associated with the equation  $[x,y]=0$.

\begin{lem}\label{lem:1.1}
Let $L$ be a  Lie ring, then $L$ is metabelian if only and if $(u,v,x,y)=(u,v,y,x)$ for all elements
$u,v,x,y$ in $L$.
\end{lem}

\begin{proof} If $L$ is metabelian,  then $(u,v;x,y)=0$. By Jacobi identity, we have
\begin{equation}\label{eq:1.6}
[[u,v],[x,y]]+[x,[y,[u,v]]]+[y,[[u,v],x]]=0.
\end{equation} Thus
$0=[x,[y,[u,v]]]+[y,[[u,v],x]]
=(x;u,v,y)+(u,v,x,y)
=-(u,v,y,x)+(u,v,x,y).
$ That is $(u,v,x,y)=(u,v,y,x)$.
Conversely, we notice that $(u,v,x,y)=(u,v,y,x)$ implies that $[x,[y,[u,v]]]+[y,[[u,v],x]]=0.$
 By Jacobi identity, we have $[u,v;x,y]=0$, as desired.
\end{proof}

\begin{lem}\label{lem:1.2}
Let $L$ be a Lie ring  of characteristic  not 2. If
$(x,y;x,z)=0$ for all $x,y,z\in L$, then $L$ is metabelian.
\end{lem}

\begin{proof} Assume that $x=s+t$, then
$
0=(s+t,y;s+t,z)
=(s,y;t,z)+(t,y;s,z),
$ and
\begin{equation}\label{eq:1.7}
(s,y;t,z)=-(t,y;s,z).
\end{equation} Let $a,b,c,d$ be arbitrary elements in $L$, then
$
(a,b;c,d)=-(c,d;a,b)=(a,d;c,b)=(d,a;b,c)
=-(b,a;d,c)=-(a,b;c,d).
$ Since the  characteristic of $L$ is not 2, $(a,b;c,d)=0$. Thus $L$ is metabelian.
\end{proof}

\begin{lem}\label{lem:1.3}
Let $L$ be a  Lie ring with $(x,y,x)=0$ for all $x,y\in L$.
\begin{enumerate}
  \item If the characteristic of $L$ is not 3, then $L^{3}=0$.
  \item If the characteristic of $L$ is  3, then $L$ is metabelian.
\end{enumerate}
\end{lem}

\begin{proof}  (1)
 Assume that $x,y,z$ are arbitrary  elements in $L$, then
$$0=(x+z,y,x+z)=(x,y,z)+(z,y,x).$$ Thus
\begin{equation}\label{eq:1.8}
(x,y,z)=-(z,y,x).
\end{equation}
By Jacobi identity, we have
\begin{eqnarray*}
(x;y,z)&=& -(y;z,x)-(z;x,y)=(y,z,x)+(z,x,y)\\
&=&-(x,z,y)-(x,z,y)=-2(x,z,y)\\
&=&2(y,z,x)=-2(x;y,z).
\end{eqnarray*} Notice that the characteristic of $L$ is not 3, so $(x;y,z)=0$ and $L^{3}=0$.

(2) Let $x,y,z$ are any elements in $L$, then
$(x,z;x,y)=(x,y,z,x)=-(x,z;x,y).$ Since the characteristic of $L$ is 3,  $(x,z;x,y)=0$. It follows from
Lemma \ref{lem:1.2}  that $L$ is metabelian.
\end{proof}

\begin{lem}\label{lem:1.4}
Let $S_{n}$ be the symmetric group of degree $n$ and $L$ be a Lie ring. Then the commutator of
 $(y_{1},\cdots,y_{n},x)$ can be expressed as the sum of $2^{n-1}$ commutators of the form
 $\pm(x,y_{\pi(1)},\cdots,y_{\pi(n)})$, where $\pi\in S_{n}$.
\end{lem}

\begin{proof} Induction on $n$, see (\cite{Sua2007}, Lemma 3.1).
\end{proof}

\section{Varieties of Nilpotent and Solvable Lie Rings}
\setcounter{equation}{0}
\renewcommand{\theequation}
{2.\arabic{equation}}

\setcounter{theorem}{0}
\renewcommand{\thetheorem}
{2.\arabic{theorem}}

In this section, we discuss the following three varieties of Lie ring  associated  with equations
\begin{eqnarray*}
(x,y_{1},\cdots,y_{n},x)&=&0;\\
(x,y;x_{1},x_{2},\cdots,x_{2n-1},x_{2n};x,z)&=&0;\\
(x,y;x_{1},x_{2},\cdots,x_{2n};x,z;y_{1},y_{2},\cdots,y_{m})&=&0.
\end{eqnarray*}

\begin{thm}\label{th:2.1}
Let $L$ be a Lie ring with $(x,y_{1},\cdots,y_{n},x)=0$.
\begin{enumerate}
  \item If the characteristic of $L$ is not 3, then $L$ is nilpotent of exponent at most $n+1$.
  \item If the characteristic of $L$ is  3, then $L$ is nilpotent of exponent at most $n+2$.
\end{enumerate}
\end{thm}

\begin{proof} (1) The proof is same as that of Theorem 3.1 in \cite{Sua2007}.

(2) We prove this  statement by induction on $n$. If the characteristic of $L$ is  3 and $n=1$, then  $(x,y_{1},x)=0$ implies that $L$ is metabelian by
Lemma \ref{lem:1.3}. Moreover, it follows from Lemma \ref{lem:1.1} that
$(u,v,x,y)=(u,v,y,x)$ for all elements $u,v,x,y$ in $L$. By (\ref{eq:1.8}), $$(x,(u,v),y)=(y,(u,v),x)=-(x,(u,v),y).$$
The characteristic of $L$ is not 2, so $(u,v,x,y)=0$. This means that $L$ is nilpotent of exponent at most $3$.
Now we suppose that $L$ is a Lie ring with $(x,y_{1},\cdots,y_{n-1},y_{n},x)=0$. Let $x=u+v$ for some elements $u,v\in L$, then
$
0=(u+v,y_{1},\cdots,y_{n},u+v)
=(u,y_{1},\cdots,y_{n},v)+(v,y_{1},\cdots,y_{n},u).
$ It follows from  Lemma \ref{lem:1.4} that
\begin{eqnarray*}
(u,y_{1},\cdots,y_{n},v)&=&-(v,y_{1},\cdots,y_{n},u)\\
&=&\pm\left[\sum_{j=1}^{2^{n-1}}(y_{n},y_{1_{j},},\cdots,y_{n_{j},}),u\right]\\
&=&\pm\sum_{j=1}^{2^{n-1}}(y_{n},y_{1_{j},},\cdots,y_{n_{j},},u),
\end{eqnarray*}where $y_{1_{j},},\cdots,y_{n_{j},}$ are some permutations of $v,y_{1},\cdots,y_{n-1}$ respectively. Thus
$$(u,y_{1},\cdots,y_{n-1},u,v)=\pm\sum_{j=1}^{2^{n-1}}(u,y_{1_{j},},\cdots,y_{n_{j},},u)=0.$$
Namely, $(u,y_{1},\cdots,y_{n-1},u)$ belongs to $C(L)$, the center of $L$. So the factor ring $L/C(L)$ satisfies the induction hypothesis, and
it is nilpotent of exponent at most $n+1$. Thus for all $x,y_{1},\cdots,y_{n},y$ in $L/C(L)$, we have
$(x,y_{1},\cdots,y_{n},y)\in C(L)$ and $(x,y_{1},\cdots,y_{n},y,z)=0$ for each $z\in L$.
Therefore $L$ is nilpotent of exponent at most $n+2$.
\end{proof}

\begin{lem}\label{lem:2.2}
Let $S_{2n}$ be the symmetric group of degree $2n$, $n\geq1$ and $L$ a Lie ring. Then the commutator of
 $(x_{1},x_{2};x_{3},x_{4};\cdots;x_{2n-1},x_{2n})$ can be expressed as the sum of $2^{n-1}$ commutators of the form
 $\pm(x_{\pi(1)},x_{\pi(2)},\cdots,x_{\pi(2n)})$, where $\pi\in S_{2n}$.
\end{lem}

\begin{proof} We prove this result by induction on $n$. When $n=2$,  the Jacobi identity implies that
\begin{equation}
(x_{1},x_{2};x_{3},x_{4})=(x_{1},x_{2},x_{3},x_{4})-(x_{1},x_{2},x_{4},x_{3}).
\end{equation}The result holds.
When $n\geq2$, we let $z=(x_{1},x_{2};x_{3},x_{4};\cdots;x_{2n-3},x_{2n-2})$. Then
\begin{eqnarray*}
[z,[x_{2n-1},x_{2n}]]&=&-[x_{2n-1},[x_{2n},z]]-[x_{2n},[z,x_{2n-1}]]\\
&=& (z,x_{2n-1},x_{2n})-(z,x_{2n},x_{2n-1}).
\end{eqnarray*}
By induction hypothesis, we assume that
$$z=\pm\sum_{j=1}^{2^{n-2}}(x_{1},x_{2},x_{\pi_{j}(1)},x_{\pi_{j}(2)},\cdots,x_{\pi_{j}(2n-2)}).$$ Then
\begin{eqnarray*}
(z,x_{2n-1},x_{2n})&=&\pm\sum_{j=1}^{2^{n-2}}(x_{1},x_{2},x_{\pi_{j}(1)},x_{\pi_{j}(2)},\cdots,x_{\pi_{j}(2n-2)},x_{2n-1},x_{2n})\\
(z,x_{2n},x_{2n-1})&=&\pm\sum_{j=1}^{2^{n-2}}(x_{1},x_{2},x_{\pi_{j}(1)},x_{\pi_{j}(2)},\cdots,x_{\pi_{j}(2n-2)},x_{2n},x_{2n-1}).
\end{eqnarray*}
Thus
$$(z;x_{2n-1},x_{2n})=\pm\sum_{j=1}^{2^{n-1}}(x_{1},x_{2},x_{\pi_{j}(1)},x_{\pi_{j}(2)},\cdots,x_{\pi_{j}(2n)})$$
as desired.
\end{proof}

\begin{thm}\label{thm:2.3}
Let $L$ be a Lie ring with $(x,y;x_{1},x_{2},\cdots,x_{2n-1},x_{2n};x,z)=0$.
Then $L^{2}$ is nilpotent of exponent at most $n+1$.
\end{thm}

\begin{proof}
Let $u=(s,y),x'_{i}=(x_{2i-1},x_{2i})$ and $v=(t,z)$. It follows from  Lemma \ref{lem:1.4} that
\begin{eqnarray*}
(s,y;x_{1},x_{2};x_{3},x_{4};\cdots;x_{2n-1},x_{2n};t,z)
&=&(u,x'_{1},x'_{2},\cdots,x'_{n},v)\\
&=& \pm\left[\sum_{j=1}^{2^{n-1}}(x'_{n},x'_{1_{j}},\cdots,x'_{n_{j}}),v\right]\\
&=& \pm\sum_{j=1}^{2^{n-1}}(x'_{n},x'_{1_{j}},\cdots,x'_{n_{j}},v),
\end{eqnarray*}
where $x'_{1_{j},},\cdots,x'_{n_{j},}$ are some permutations of $u,x'_{1},\cdots,x'_{n-1}$ respectively. Let $x'_{n}=v=(t,z)$.
By Lemma \ref{lem:2.2}, we obtain
\begin{eqnarray*}
(s,y;x_{1},x_{2};x_{3},x_{4};\cdots;x_{2n-1},x_{2n};t,z)
&=& \pm\sum_{j=1}^{2^{n-1}}(t,z;x'_{1_{j}},\cdots,x'_{n_{j}};t,z)\\
&=& \pm\sum_{j=1}^{2^{n-1}}(t,z;x_{1_{j}},x_{2_{j}},\cdots,x_{n_{j}};t,z)=0.
\end{eqnarray*}
Hence $L^{2}$ is nilpotent of exponent at most $n+1$.
\end{proof}

\begin{thm}\label{thm:2.4}
Let $L$ be a Lie ring with $(x,y;x_{1},x_{2},\cdots,x_{2n};x,z;y_{1},y_{2},\cdots,y_{m})=0,$ where $2n\geq m\geq 1$ and $n\geq 1$.
Then $L$ is solvable of length at most $n+1$.
\end{thm}

\begin{proof}
We write $C_{L}(L^{m})$ for the centralizer of $L^{m}$ in $L$. The factor ring $L/C_{L}(L^{m})$ satisfies the identity
$(x,y;x_{1},x_{2},\cdots,x_{2n};x,z)=0$, it follows  from Theorem \ref{thm:2.3} that $(L/C_{L}(L^{m}))^{2}$ is nilpotent of exponent at most $n+1$.
Thus
$$(s,y;x_{1},x_{2};x_{3},x_{4};\cdots;x_{2n-1},x_{2n};t,z)\in C_{L}(L^{m})$$
and $(s,y;x_{1},x_{2};x_{3},x_{4};\cdots;x_{2n-1},x_{2n};t,z;y_{1},y_{2},\cdots,y_{m})=0.$  From  Lemma \ref{lem:2.2}, we know that
each element $u$ in $L^{(n+2)}$ can be expressed  as a finite sum of the
form $$\pm\left(x_{\pi(1)},x_{\pi(2)};x_{\pi(3)},x_{\pi(4)};\cdots;x_{\pi(2^{n+2}-1)},x_{\pi(2^{n+2})}\right),$$ where $\pi\in S_{2^{n+2}}$. That is
$$u=\pm\sum(x_{\pi(1)},x_{\pi(2)};\cdots;x_{\pi(2n+3)},x_{\pi(2n+4)};x_{\sigma(2n+5)},x_{\sigma(2n+6)},\cdots,x_{\sigma(2^{n+2})}),$$ for some
$\pi,\sigma \in S_{2^{n+2}}$. We note that $2^{n+2}-(2n+4)\geq m$.
Hence
$L^{(n+2)}=\{0\}$, as desired.
\end{proof}

\begin{prop}
The variety $V_{n}$ associated with the equation
 $$(x_{1},\cdots,x_{n},x_{n+1},x_{n+2})=(x_{1},\cdots,x_{n},x_{n+2},x_{n+1})$$ is precisely the variety
 $V'_{n}$ associated with the equation
$(x_{1},\cdots,x_{n};x_{n+1},x_{n+2})=0.$ In particular, $V_{2}$ is the variety of metabelian Lie rings.
\end{prop}

\begin{proof}  Let $z=(x_{1},\cdots,x_{n})$. This statement follows immediately  from Jacobi identity.
\end{proof}

\begin{thm}
The variety $P_{n}$ associated with the equation
 $$(x_{1},\cdots,x_{i},x_{i+1},\cdots, x_{n})=(x_{1},\cdots,x_{i+1},x_{i},\cdots, x_{n})$$ is precisely the variety
 $P'_{n}$ associated with the equation
 $(x_{1},\cdots,(x_{i},x_{i+1}),\cdots, x_{n})=0.$ Moreover, each Lie ring $L$ in $P_{n}$ is solvable.
\end{thm}

\begin{proof}
 Let  $a=(x_{1},\cdots,x_{i-1}),b=(x_{i+2},\cdots,x_{n})$. It follows from Jacobi identity that
\begin{eqnarray*}
0&=&((a,(x_{i},x_{i+1}))+(x_{i},(x_{i+1},a))+(x_{i+1},(a,x_{i})),b)\\
&=& ((a,(x_{i},x_{i+1}))+((a,x_{i+1}),x_{i})-((a,x_{i}),x_{i+1}),b)\\
&=& (a,(x_{i},x_{i+1}),b)+(a,x_{i+1},x_{i},b)-(a,x_{i},x_{i+1},b).
\end{eqnarray*} Thus the equation $(x_{1},\cdots,x_{i},x_{i+1},\cdots, x_{n})=(x_{1},\cdots,x_{i+1},x_{i},\cdots, x_{n})$ is equivalent to
$$(x_{1},\cdots,(x_{i},x_{i+1}),\cdots, x_{n})=0.$$
It is easy to see that $P_{2}$ is the variety of abelian Lie rings. We use induction on $n$ to prove the solvability of a Lie ring $L$ in $P_{n}$.
Since $(x_{1},\cdots,(x_{i},x_{i+1}),\cdots, x_{n})=0$,  $$(x_{1},\cdots,(x_{i},x_{i+1}),\cdots, x_{n-1})(i< n-1)$$ belongs to $C(L)$, the center of $L$. That
is $(x_{1},\cdots,(x_{i},x_{i+1}),\cdots, x_{n-1})=0$ in $L/C(L)$. By induction hypothesis, we conclude that the factor ring  $L/C(L)$ is solvable,
and thus  $L$ is also solvable.
If $(x_{1},\cdots,x_{n-2},(x_{n-1},x_{n}))=0$, we observe that the equations $(x_{n-1},x_{n};x_{1},\cdots,x_{n-2})=0$  and
$(x_{1},x_{2};x_{3},\cdots, x_{n})=0$ are equivalent. Hence $L$ is also solvable.
\end{proof}

\begin{lem}\label{lem:2.7}
Let $L$ be a Lie ring with $(x_{1},x_{2},\cdots,x_{n};y_{1},y_{2})=0$, where $n\geq2$. If $L/C(L)$ satisfies
$C(n+1)$, then $L$ satisfies $C(n+2)$.
\end{lem}

\begin{proof} By Jacobi identity, we see that
\begin{eqnarray}\label{eq:2.2}
[a,[x,y]]=0 \Leftrightarrow (a,x,y)=(a,y,x).
\end{eqnarray}
Thus $L$ satisfies $C(n+2,\rho_{1})$ for $\rho_{1}=(n+1~n+2)$. Since $L/C(L)$ satisfies
$C(n+1)$,  $L/C(L)$ satisfies
$C(n+1,\rho_{2})$ for $\rho_{2}=(34\cdots n+1)$. Namely,
$$((x_{1},x_{2},x_{3},\cdots,x_{n+1})-(x_{1},x_{2},x_{4},\cdots,x_{n+1},x_{3}),x_{n+2})=0.$$
Thus $L$ satisfies $C(n+2,\rho_{2})$. Notice that $L$ satisfies $C(n+2,\rho_{1})$ and $\rho_{1},\rho_{2}$ generate the group
of all permutations of $\{3,4\cdots, n+2\}$, thus $L$ satisfies $C(n+2)$.
\end{proof}

\begin{thm}\label{thm:2.8}
Let $L$ be a Lie ring, then  $L$ satisfies
\begin{equation}\label{eq:2.3}
(x_{1},x_{2},\cdots,x_{n})=(x_{1},x_{\rho(2)},\cdots,x_{\rho(n)})
\end{equation}
  for all nonidentical permutations $\rho$
of $\{2,\cdots,n\}$,
if and only if $L$ is nilpotent of exponent at most $n-1$.
\end{thm}

\begin{proof} We shall use the induction on $n$. For $n=3$, $(x_{1},x_{2},x_{3})=(x_{1},x_{3},x_{2})$ if and only if
$[x_{1},[x_{2},x_{3}]]=0$ by (\ref{eq:2.2}). Assume that $L$ satisfies the hypothesis for $n>3$.  By (\ref{eq:2.2}), we deduce that
$L$ satisfies $(x_{1},x_{2},\cdots;x_{n-1},x_{n})=0$. Notice that $L$ satisfies (\ref{eq:2.3}) for all nonidentical permutations $\rho$.
A similar argument as in the proof of
Lemma \ref{lem:2.7} implies that $L/C(L)$  satisfies (\ref{eq:2.3}) with $n$ replaced by $n-1$ for any permutation $\sigma$ of $\{2,\cdots,n-1\}$.
Induction  hypothesis yields that $L/C(L)$ is nilpotent of exponent at most $n-2$, and $L$ is nilpotent of exponent at most $n-1$.
The proof of the converse is trivial.
\end{proof}

\section{Varieties Defined by Properties of Subrings}
\setcounter{equation}{0}
\renewcommand{\theequation}
{3.\arabic{equation}}

\setcounter{theorem}{0}
\renewcommand{\thetheorem}
{3.\arabic{theorem}}

In this section, we return to some of Macdonald's ideas in \cite{Mac1962} which were originally used to discuss the group-theoretic problems.

\begin{thm} Let $L$ be a Lie ring.
\begin{enumerate}
  \item  If $L$ has characteristic not 3, then every $n$-generated  subring of $L$ is  nilpotent of exponent at most $n$ if and only if $L$ is
also nilpotent of exponent at most $n$.
  \item  If $L$ has characteristic  3, then every $n$-generated  subring of $L$ is  nilpotent of exponent at most $n$ if and only if $L$ is
also nilpotent of exponent at most $n+1$.
\end{enumerate}
\end{thm}

\begin{proof}
$(\Rightarrow)$ Let  $x,y,x_{1},\cdots,x_{n-2}$ be any elements of $L$. The subring $\langle x,y,x_{1},\cdots,x_{n-2}\rangle$
is  nilpotent of exponent at most $n$. Let $a=(x_{1},\cdots,x_{n-2})$, then
\begin{eqnarray*}
(x,y,x,a)&=&-(a,(x,y),x)-(x,a;x,y)\\
&=&(a,x;x,y)-(a,(x,y),x)\\
&=&(a,x,x,y)-(a,x,y,x)-((-(x;y,a)-(y;a,x)),x)\\
&=&(a,x,x,y)-(a,x,y,x)+(a,y,x,x)-(a,x,y,x)=0.
\end{eqnarray*} Namely  $(x,y,x;x_{1},\cdots,x_{n-2})=0$, and thus $(x,y,x)$ in $C_{L}(L^{n-2})$, the centralizer of $L^{n-2}$ in $L$.
Notice that $L/C_{L}(L^{n-2})$ satisfies the conditions in Lemma \ref{lem:1.3}. It follows that $L/C_{L}(L^{n-2})$ has exponent no more than 3 in case of characteristic not 3,  and
is metabelian in case of characteristic  3. Therefore $L$ is nilpotent of exponent at most $n$ (char$L\neq 3$) and $n+1$ (char$L=3$).

$(\Leftarrow)$ Obviously.
\end{proof}

This results with Theorem \ref{th:2.1} give the following conclusions.

\begin{coro} Let $\mathfrak{L}$ be the variety of Lie rings.
\begin{enumerate}
  \item If each Lie ring in $\mathfrak{L}$ is of characteristic not 3 and with the property that every $n$-generated  subring  is  nilpotent of exponent at most $n$, then $\mathfrak{L}$ is determined  by the equation
$$(x,x_{1},\cdots,x_{n-1},x)=0.$$
  \item If each Lie ring in $\mathfrak{L}$ is of characteristic  3 and with the property that every $n$-generated  subring  is  nilpotent of exponent at most $n+1$, then $\mathfrak{L}$ is determined  by the equation
$$(x,x_{1},\cdots,x_{n-1},x)=0.$$
\end{enumerate}
\end{coro}

With an analogous argument, we have

\begin{thm}
The variety of Lie rings of characteristic not 2 in which every  subring with $2^{n+1}-1$ generators  is  solvable of length at most $n$, is determined by the equation
$$(x,y;x_{1},x_{2},\cdots,x_{2n-2};x,z;y_{1},y_{2},\cdots,y_{2n-2})=0,$$
 where  $n\geq 2$.
\end{thm}

\begin{proof} By Theorem \ref{thm:2.4}, it suffices to show that
a Lie ring $L$ with every  subring with $2^{n+1}-1$ generators  is  solvable of length at most $n$, is solvable of length at most $n$.
If $n=2$, then for all  $x,y,x_{1},x_{2},x,z,y_{1},y_{2}$ in $L$,
$((x,y;x_{1},x_{2}),(x,z;y_{1},y_{2}))=0.$ Let $z=y$, then $$((x,y;x_{1},x_{2}),(x,y;y_{1},y_{2}))=0.$$ That means that
$L^{(1)}$ satisfies the condition of Lemma \ref{lem:1.2}, so $L^{(1)}$ is metabelian. Thus $L^{(3)}=0$ and $L$ is  solvable of length at most $2$.

Assume that $x,y,x_{1},x_{2},\cdots,x_{2^{n}-2},x,z,y_{1},y_{2},\cdots,y_{2^{n}-2}$ are $2^{n+1}-1$ arbitrary  elements of $L$. Since
the subring $\langle x,y,x_{1},x_{2},\cdots,x_{2^{n}-2},x,z,y_{1},y_{2},\cdots,y_{2^{n}-2}\rangle$ is solvable of length at most $n$,
$$[\cdots[[x,y],[x_{1},x_{2}]]\cdots],\cdots,x_{2^{n}-2}]; [\cdots[[x,z],[y_{1},y_{2}]]\cdots],\cdots,y_{2^{n}-2}]\cdots]=0.$$
We let $z=y$, then
$$[\cdots[[x,y],[x_{1},x_{2}]]\cdots],\cdots,x_{2^{n}-2}]; [\cdots[[x,y],[y_{1},y_{2}]]\cdots],\cdots,y_{2^{n}-2}]\cdots]=0.$$
That is, in $L^{(1)}$,  every  subring with $2^{n}-1$ generators  is  solvable of length at most $n-1$. By induction hypothesis,
it follows that $L^{(1)}$ is  solvable of length at most $n-1$. Thus $L^{(n+1)}=0$, as desired.
\end{proof}

\section*{\textit{Acknowledgments}}
The authors are grateful to
an anonymous referee for his (or her) valuable
 comments and suggestions on the
 first version of the article.
This work was supported by the Fundamental Research Funds for the Central Universities (No. 111494343) and NSF of China (No.11026136).

\end{document}